\documentclass[runningheads]{llncs}
\usepackage[T1]{fontenc}
\usepackage{graphicx}
\usepackage{proof}
\usepackage{color}
\usepackage{adjustbox}
\usepackage{amsfonts}
\usepackage{comment}

\newcommand{\FOLID}{\mathrm{FOL}_{\mathrm{ID}}}
\newcommand{\LKomega}{\mathrm{LKID}^\omega}
\newcommand{\CLK}{\mathrm{CLKID}^\omega}
\newcommand{\CSLomega}{\mathrm{CSL}^\omega}

\newcommand{\pr}{{\rm Pr}}
\newcommand{\prplus}{\mathrm{Pr}_{+}}
\newcommand{\prvar}{\mathrm{Pr}_{+}^\mathrm{var}}
\newcommand{\prminus}{\mathrm{Pr}_{-}}
\newcommand{\promega}{\mathrm{Pr}^{\omega}}
\newcommand{\prd}[1]{\mathrm{Pr}_{#1}^\omega}

\newcommand{\f}{f_-}

\newcommand{\fomega}{f^\omega}
\newcommand{\fd}[1]{f_{#1}^\omega}

\newcommand{\FV}{{\rm FV}}
\newcommand{\Ent}{{\rm Seq}}
\newcommand{\prsubst}{{\Theta}}
\newcommand{\psclosure}{{\mathcal{C}_{ps}}}
\newcommand{\dom}{{\rm dom}}
\newcommand{\img}{{\rm img}}

\newcommand{\axiom}{\mathrm{Axiom}}
\newcommand{\wk}{\mathrm{Wk}}
\newcommand{\cut}{\mathrm{Cut}}
\newcommand{\subst}{\mathrm{Subst}}
\newcommand{\notL}{\lnot \mathrm{L}}
\newcommand{\notR}{\lnot \mathrm{R}}
\newcommand{\veeL}{\vee \mathrm{L}}
\newcommand{\veeR}{\vee \mathrm{R}}
\newcommand{\wedgeL}{\wedge \mathrm{L}}
\newcommand{\wedgeR}{\wedge \mathrm{R}}
\newcommand{\toL}{\to \mathrm{L}}
\newcommand{\toR}{\to \mathrm{R}}
\newcommand{\forallL}{\forall \mathrm{L}}
\newcommand{\forallR}{\forall \mathrm{R}}
\newcommand{\existsL}{\exists \mathrm{L}}
\newcommand{\existsR}{\exists \mathrm{R}}
\newcommand{\eqL}{= \mathrm{L}}
\newcommand{\eqR}{= \mathrm{R}}
\newcommand{\ul}{\mathrm{UL}}
\newcommand{\ur}{\mathrm{UR}}

\newcommand{\Np}{\mathrm{N}}
\newcommand{\Op}{\mathrm{O}}
\newcommand{\Ep}{\mathrm{E}}

\newcommand{\system}{\mathcal{S}}

\begin{document}
\title{Admissibility of Substitution Rule in Cyclic-Proof Systems}
%
%
\author{Kenji Saotome\inst{1}\orcidID{0009-0008-6917-5673} \and Koji Nakazawa\inst{1}\orcidID{0000-0001-6347-4383}}
\authorrunning{K.~Saotome and K.~Nakazawa}
%
\institute{Nagoya University, Furo-cho, Chikusa-ku, Nagoya 464-8601, Japan\\
\email{\{saotomekenji, knak\}@sqlab.jp}
}
\maketitle              
\begin{abstract}
This paper investigates the admissibility of the substitution rule in cyclic-proof systems.
The substitution rule complicates theoretical case analysis and increases computational cost in proof search since every sequent can be a conclusion of an instance of the substitution rule; hence, admissibility is desirable on both fronts.
While admissibility is often shown by local proof transformations in non-cyclic systems, such transformations may disrupt cyclic structure and do not readily apply.
Prior remarks suggested that the substitution rule is likely non-admissible in the cyclic-proof system $\CLK$ for first-order logic with inductive predicates. 
In this paper, we prove admissibility in $\CLK$, assuming the presence of the cut rule.
Our approach unfolds a cyclic proof into an infinitary form, lifts the substitution rules, and places back edges to construct a cyclic proof without the substitution rule.
If we restrict substitutions to exclude function symbols, the result extends to a broader class of systems, including cut-free $\CLK$ and cyclic-proof systems for the separation logic.

\end{abstract}
\section{Introduction}
The framework of \emph{cyclic proofs}~\cite{Brotherston06,Brotherston11,Doumane17} extends traditional proof trees by introducing cyclic structures to reasoning about inductively defined predicates. 
To ensure soundness, only those proofs that satisfy the \emph{global trace condition} are admitted. 
The condition guarantees that every infinite path makes progress infinitely often by unfolding an inductive predicate.
For example, the derivation
\begin{center}
    \includegraphics[width = .375\columnwidth]{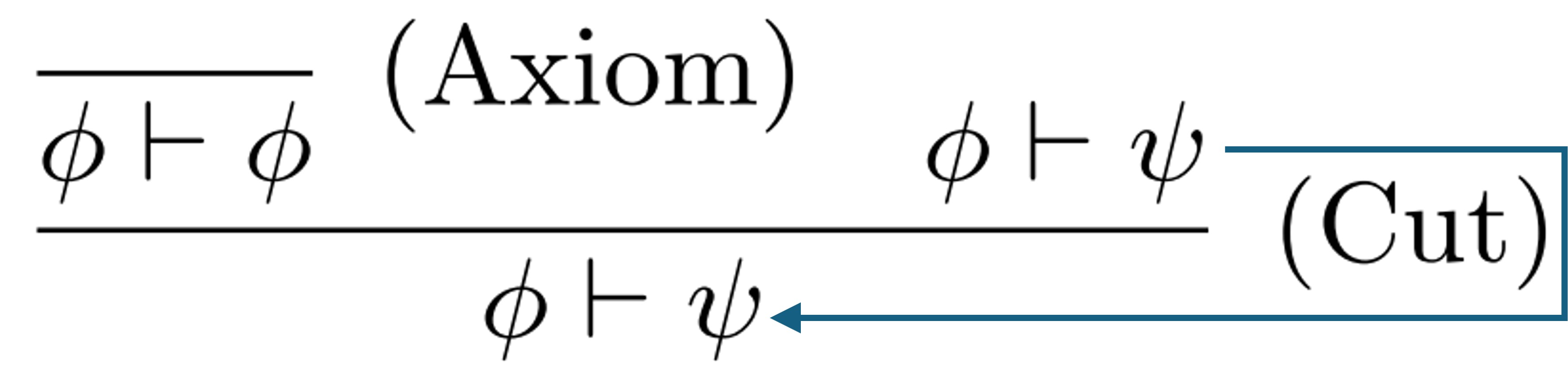}
\end{center}
fails to satisfy the global trace condition, and the conclusion is invalid in general.
On the other hand, the proof shown in Fig.~\ref{fig: N E O} is a
cyclic proof in $\CLK$~\cite{Brotherston06}, a cyclic-proof system for the first-order logic with inductive predicates.
\begin{figure}
    \centering
    \includegraphics[width=0.75\linewidth]{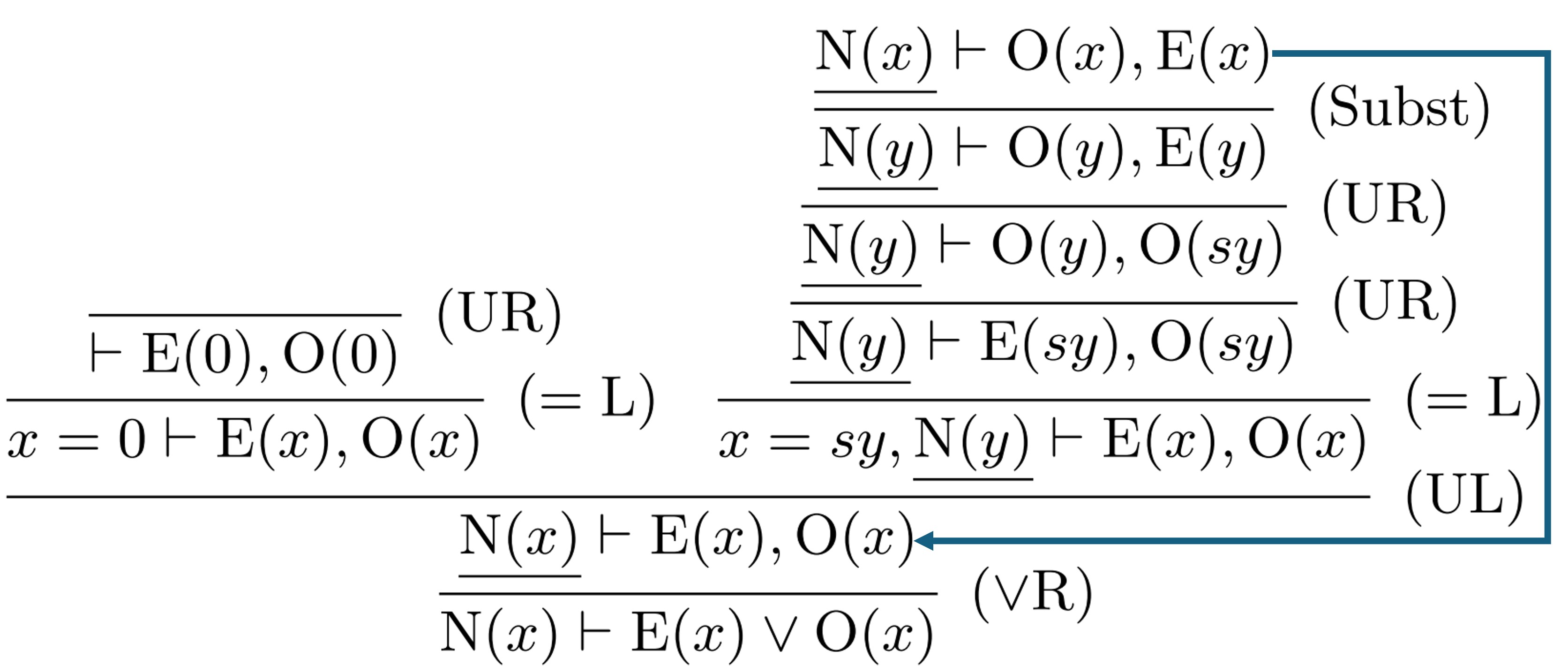}
    \caption{A cyclic proof of $\Np(x)\vdash\Ep(x)\vee\Op(x)$}
    \label{fig: N E O}
\end{figure}
Intuitively, $\Np(x)$, $\Ep(x)$, and $\Op(x)$ respectively represent that $x$ is a natural number, an even number, and an odd number.
The leaf node $\Np(x)\vdash \Op(x),\Ep(x)$ is called \emph{bud}, and the internal node associated with it is called its \emph{companion}. 
A cycle is formed through the connection between a bud and its companion.
This proof satisfies the global trace condition since every infinite path makes progress infinitely often by unfolding the underlined inductive predicate $\Np$ in the antecedents at $(\ul)$.

When we search for proof from the conclusion to premises, unlike the ordinary induction rule, cyclic
proofs may be found without a prior choice of the induction hypothesis.
Consequently, avoiding heuristic searches for induction hypotheses, cyclic proofs have been investigated in the context of proof search~\cite{Brotherston11b,Brotherston12,Chu15,Ta16,Tatsuta19}. 

The substitution rule
\[
\infer[(\subst)]{\Gamma[\theta]\vdash\Delta[\theta]}{\Gamma\vdash\Delta}
\]
is often introduced in proof systems.
Some cyclic-proof systems contain the substitution rule as an explicit inference rule~\cite{Brotherston06,Zhang22,Masuoka23}.
On the other hand, some cyclic-proof systems implicitly introduce substitutions into the bud-companion correspondence~\cite{Ulrich02,Kimura19,Saotome24}.
That is, instead of requiring that the bud and its companion be syntactically identical, they require that the bud be a substitution instance of the companion. 
The substitution rule helps to construct cyclic structures.
In the proof in Fig.~\ref{fig: N E O}, a fresh variable $y$ is introduced by $(\ul)$, and this $y$ is unified with $x$ by $(\subst)$ to identify the sequents $\Np(x)\vdash\Ep(x),\Op(x)$ and $\Np(y)\vdash\Op(y),\Ep(y)$.

The admissibility of the substitution rule is meaningful both from theoretical and practical perspectives.
Theoretically, since the substitution rule is always applicable, it complicates proof analysis.  Consequently, establishing the admissibility of the substitution rule enables a more tractable proof
analysis.
On the practical side, when the substitution rule is admissible, implementations are afforded greater flexibility in how and when to apply substitutions.  
For example, one might reserve the application of the substitution rule for specific moments, such as when considering bud-companion relationships, when merging with already-proven theorems, or even choose never to apply substitution at all.  
Since many substitutions can be applied at each step, unrestricted application of the substitution rule increases computational cost.

In many non-cyclic-proof systems, the substitution rule is admissible. 
This is achieved by (i) lifting the substitution rule and (ii) eliminating the substitution rule. 
In the lifting phase (i), one transforms a derivation tree
{\small
\[
\infer[(\subst)]{\Gamma[\theta]\vdash\Delta[\theta]}
    {\infer[(R)]{\Gamma\vdash\Delta}{\Gamma_1\vdash\Delta_1 &\cdots
    &\Gamma_n\vdash\Delta_n}}
\mbox{{\normalsize\quad into\quad }}
\infer[(R)]{\Gamma[\theta]\vdash\Delta[\theta]}
    {\infer[(\subst)]{\Gamma_1[\theta_1]\vdash\Delta_1[\theta_1]}{\Gamma_1\vdash\Delta_1}
    &\cdots
    &\infer[(\subst)]{\Gamma_n[\theta_n]\vdash\Delta_n[\theta_n]}{\Gamma_n\vdash\Delta_n}
    }.
\]
}
 In the elimination phase (ii), one transforms a derivation tree
 {\small
\[
\infer[(\subst)]{\Gamma[\theta]\vdash\Delta[\theta]}
    {\infer[(\axiom)]{\Gamma\vdash\Delta}{}}
\mbox{{\normalsize\quad into\quad }}
\infer[(\axiom)]{\Gamma[\theta]\vdash\Delta[\theta]}
    {}.
\]    
}

In cyclic-proof systems, however, the admissibility of the substitution rule is still unclear. At least, it is not possible to remove $(\subst)$ in the same way as in non-cyclic-proof systems.
In the lifting phase, the bud-companion relationship may be destroyed if the sequent $\Gamma\vdash\Delta$, which is deleted during the transformation, is a companion.
This is not merely a temporary effect during proof transformation; for example, as illustrated in Fig.~\ref{fig: lifting subst breaking}, even if the substitutions are lifted to the bud, it may still fail to choose the corresponding companion.
In the elimination phase, it is not possible to eliminate the substitution when the leaf of a derivation tree is a bud. 
\begin{figure}
\centering
\includegraphics[width=0.7\textwidth]{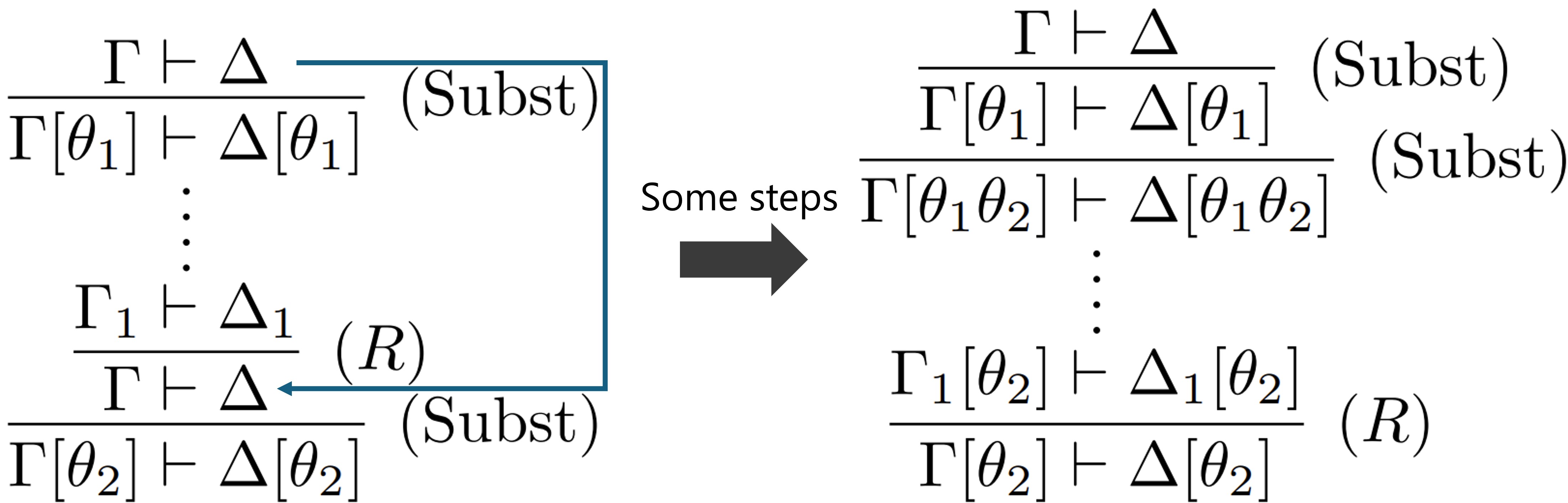}
\caption{Lifting $(\subst)$ may break bud-companion correspondence}
\label{fig: lifting subst breaking}
\end{figure}

Brotherston~\cite{Brotherston06} raised the question of whether the substitution rule is admissible in the cyclic-proof system $\CLK$, suggesting that it is likely not admissible in general, at least in the cut-free setting.

This paper proves the admissibility of the substitution rule in the cyclic-proof system $\CLK$ for the first-order logic with inductive predicates~\cite{Brotherston11}, assuming the presence of the cut rule.
Furthermore, if function symbols with positive arity are not introduced by the substitution rule, its applications can be eliminated even without the cut rule.  
This result can be applied to other proof systems as well.  
For example, in the cyclic-proof system for separation logic, where function symbols except the constant $\mathrm{nil}$ are generally not introduced, the substitution rule is admissible regardless
of the presence of the cut rule.

The proof 
\begin{center}
\includegraphics[width =0.88\textwidth]{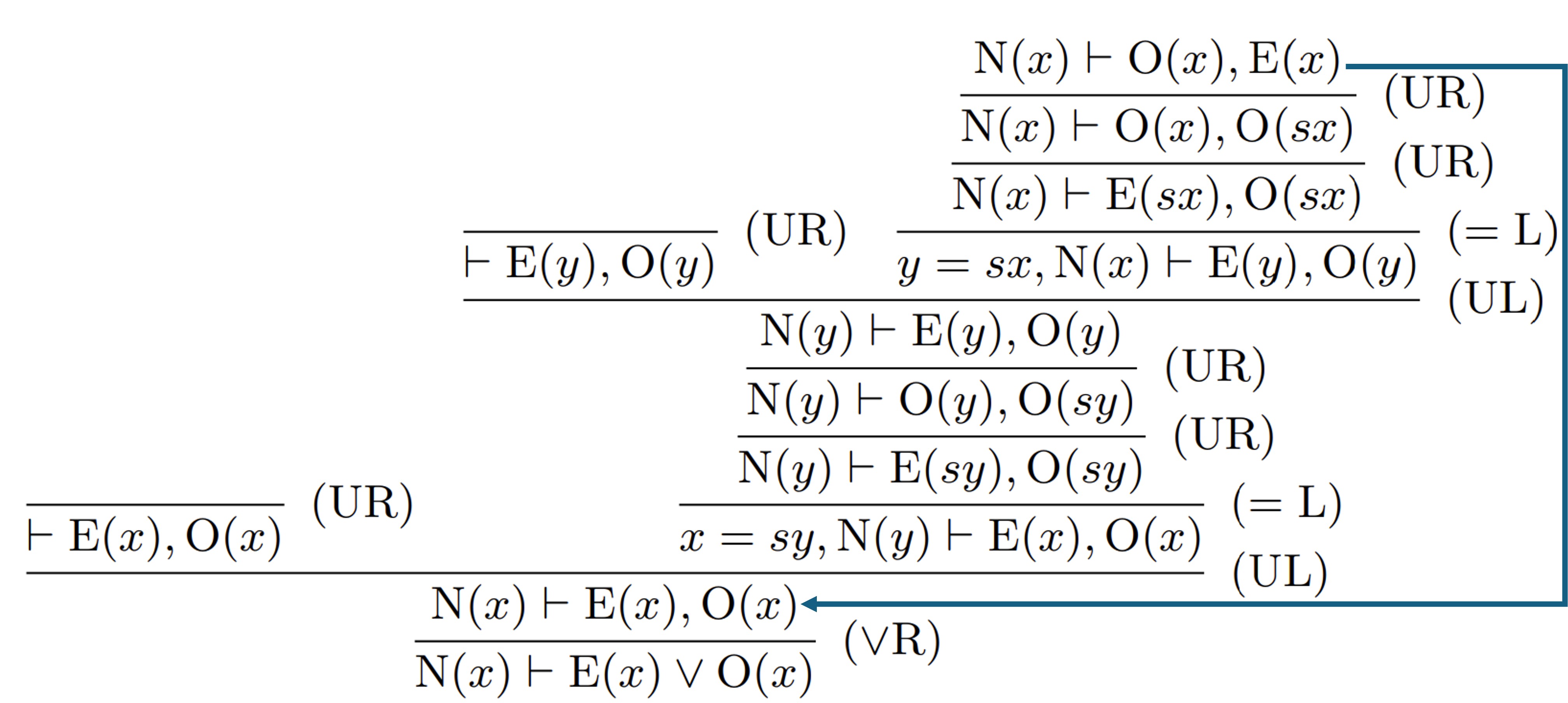}
\end{center}
is one of the cyclic proofs of $\Np(x)\vdash\Ep(x)\vee\Op(x)$ without substitution rules, that is constructed by unfolding the cyclic proof in Fig.~\ref{fig: N E O}.
Regarding the new variable $y$ introduced by the lower $(\ul)$, the original proof unifies it with the variable $x$ using $(\subst)$. 
In the new proof, this is accomplished by choosing $x$ itself as the fresh variable in the upper $(\ul)$.
Building on this observation, we consider unfolding a cyclic proof into an infinite proof, wherein substitution rules can be systematically lifted upward. 
After performing this lifting to a sufficiently great height in the infinite proof tree, we reconstruct a cyclic proof without substitutions.

To guarantee the global trace condition, we choose buds and their companions in the reconstructed cyclic proof only when they have the same entailment occurrence as in the original proof.
This careful choice prevents the emergence of unintended infinite paths.


{\bfseries Structure of the paper.} 
Section~2 introduces a cyclic-proof system $\CLK$ and an infinitary-proof system $\LKomega$ for the first-order logic $\FOLID$ with inductive predicates. 
In Section~3, we prove the admissibility of the substitution rule in $\CLK$.
In Section~4, we discuss the admissibility of the substitution rule in other proof systems.
In Section~5, we conclude.
\section{Cyclic-Proof System for First-Order Logic with Inductive Predicates}

This section recalls the first-order logic with inductive predicates, together with the cyclic- and infinitary-proof systems, as proposed by Brotherston et al.~\cite{Brotherston06,Brotherston11}.
For the unfolding rules of inductive predicates, however, we follow the style of Kimura et al.~\cite{Kimura19}, in which the rule names do not explicitly mention the inductive predicates.

\subsection{First-Order Logic $\FOLID$ with Inductive Predicates}
We define the formula of $\FOLID$.
A \emph{function symbol}, denoted by $f,g,\ldots$, has a fixed arity. A nullary function symbol, denoted by $c,c_1,\ldots$, is called a \emph{constant symbol}. 
We denote a \emph{variable} by $x,y,z,\ldots$. We use $\vec{x}$ for a finite sequence of variables. 
A \emph{term}, denoted by $t, u, \ldots$, is defined recursively as either a variable or an expression of the form $f(\vec{t})$, where $\vec{t}$ is a finite sequence of terms.
We write $t(x)$ to make explicit that the variable $x$ occurs in the term $t$; analogously, we use the notation $\vec{t}(\vec{x})$ for the same purpose.
$=$ is introduced as a special two-argument predicate symbol.
We use $P,P_1,\ldots$ for \emph{inductive predicate symbols}. 
We also use $Q, Q_1,\ldots$ for \emph{ordinary predicate symbols}. 
Each predicate symbol has a fixed arity. 

\begin{definition}[Formulas]
The formulas of $\FOLID$, denoted $\phi,\psi,\ldots$, are as follows:
\[
\phi::= P(\vec{t})\mid Q(\vec{u})
\mid t=u
\mid\lnot\phi
\mid\phi\vee\phi
\mid\phi\wedge\phi
\mid\phi\to\phi
\mid \exists x.\phi
\mid\forall x.\phi,  
\]
where the lengths of $\vec{t}$ and $\vec{u}$ are the arities of $P$ and $Q$, respectively.

Formulas in $\FOLID$ are identified up to $\alpha$-equivalence under variable binding by $\forall$ and $\exists$. For formulas $\phi$ and $\psi$ that are syntactically equal, that is, $\alpha$-equivalent,
we denote $\phi\equiv\psi$.
\end{definition}

The substitutions of terms, denoted by $\theta,\theta_1,\ldots$, are finite mappings from term variables to terms. 
The formula $\phi[\theta]$ denotes the result of the capture-avoiding substitution $\theta$
to free variables in $\phi$.
$\theta_1\theta_2$ denotes the composition of the substitutions $\theta_1$ and $\theta_2$, and it holds that $\phi[\theta_1\theta_2]\equiv(\phi[\theta_1])[\theta_2]$.
$\theta[x_1\to y_1,\ldots,x_n\to y_n]$ denotes the substitution obtained
from $\theta$ by changing the image of $x_i$ to $y_i$ for $1\leq i\leq
n$, respectively. 
For a finite set of formulas $\Gamma$, $\Gamma[\theta]$ is a finite set such that $\Gamma[\theta]=\{\phi[\theta]\mid \phi\in\Gamma\}$.

Specifically, a substitution whose image includes only variables and constants is called an \emph{atomic substitution}. On the other hand, a substitution whose image contains a term with a function symbol of positive arity is called a \emph{composite substitution}.
The domain and the image of a mapping $f$ are denoted by $\dom(f)$ and $\img(f)$, respectively. 

The definition of each inductive predicate $P(\vec{t})$ is given by a definition set.
\begin{definition}[Definition sets]
An inductive definition set $\Phi$ is a finite set of productions in the following form:
\[
\infer[]{P(\vec{t})}{Q_1(\vec{u_1})&\cdots& Q_n(\vec{u_n})&P_{1}(\vec{t_1})&\ldots&P_{m}(\vec{t_m})},
\]
where each $Q_i(1\leq i\leq n)$ is an ordinary predicate symbol and each $P_j(1\leq j\leq m)$ is an inductive predicate symbol.
\end{definition}

\begin{example}
The following are the productions for $\Np$, $\Ep$, and $\Op$, which represent natural numbers, even numbers, and odd numbers, respectively.
\[
\infer[]{\Np(0)}{},\qquad
\infer[]{\Np(sx)}{\Np(x)},\qquad
\infer[]{\Ep(0)}{},\qquad
\infer[]{\Ep(sx)}{\Op(x)}\mbox{, and}\qquad
\infer[]{\Op(sx)}{\Ep(x)}.
\]
\end{example}

We consider the \emph{standard model}~\cite{Brotherston06,Brotherston11} as the semantics.
\begin{definition}[Sequents]
Let $\Gamma$ and $\Delta$ be finite sets of formulas.
A sequent, denoted by $e$, is defined as $\Gamma\vdash\Delta$.
We say that $\Gamma\vdash\Delta$ is \emph{valid} with respect to the standard model.
\end{definition}

\subsection{Infinitary-Proof System $\LKomega$}
In this subsection, we introduce the infinitary-proof system $\LKomega$ used to establish the admissibility of the substitution rule in the cyclic-proof system $\CLK$.

First, we define the inference rules of $\LKomega$. 
These inference rules are also used in the cyclic-proof system $\CLK$.
\begin{definition}[Inference rules]
The inference rules of $\LKomega$ and $\CLK$ except for inductive predicates unfolding are as follows:
\[
\infer[\Gamma\cap\Delta\neq\emptyset (\axiom)]{\Gamma\vdash\Delta}{}\qquad
\infer[\Gamma'\subseteq \Gamma,\Delta'\subseteq \Delta(\wk)]{\Gamma\vdash\Delta}{\Gamma'\vdash\Delta'}
\]
\[
\infer[(\cut)]{\Gamma\vdash\Delta}{\Gamma\vdash\phi,\Delta &\Gamma,\phi\vdash\Delta}\quad
\infer[(\subst)]{\Gamma[\theta]\vdash\Delta[\theta]}{\Gamma\vdash\Delta}\quad
\infer[(\notL)]{\Gamma,\lnot\phi\vdash\Delta}{\Gamma\vdash\phi,\Delta}\quad
\infer[(\notR)]{\Gamma\vdash\Delta,\lnot\phi}{\Gamma,\Delta\vdash\phi}
\]
\[
\infer[(\veeL)]{\Gamma,\phi\vee\psi\vdash\Delta}
    {\Gamma,\phi\vdash\Delta &\Gamma,\psi\vdash\Delta} \qquad
\infer[(\veeR)]{\Gamma\vdash\Delta,\phi\vee\psi}
    {\Gamma\vdash\Delta,\phi,\psi}
\]
\[
\infer[(\wedgeL)]{\Gamma,\phi\wedge\psi\vdash\Delta}
    {\Gamma,\phi,\psi\vdash\Delta} \qquad
\infer[(\wedgeR)]{\Gamma\vdash\Delta,\phi\wedge\psi}
    {\Gamma\vdash\Delta,\phi &\Gamma\vdash\Delta,\psi}
\]
\[
\infer[(\toL)]{\Gamma,\phi\to\psi\vdash\Delta}
    {\Gamma\vdash\phi,\Delta &\Gamma,\psi\vdash\Delta} \qquad
\infer[(\toR)]{\Gamma\vdash\Delta,\phi\vee\psi}
    {\Gamma\phi\vdash\psi\Delta}
\]
\[
\infer[(\forallL)]{\Gamma,\forall x.\phi\vdash\Delta}{\Gamma,\phi[t/x]\vdash\Delta}\quad
\infer[(\forallR)]{\Gamma\vdash\forall x.\phi,\Delta}{\Gamma\vdash\phi[y/x],\Delta}\quad
\infer[(\existsL)]{\Gamma,\exists x.\phi\vdash\Delta}{\Gamma,\phi[y/x]\vdash\Delta}\quad
\infer[(\existsR)]{\Gamma\vdash\exists x.\phi,\Delta}{\Gamma\vdash\phi[t/x],\Delta}
\]
\[
\infer[(\eqL)]{\Gamma[t/x,u/y],t=u\vdash\Delta[t/x,u/y]}
    {\Gamma[u/x,t/y]\vdash\Delta[u/x,t/y]}\qquad
\infer[(\eqR)]{\Gamma\vdash t=t,\Delta}{},
\]
where $t$ in $(\existsL)$ and $(\exists)$ is arbitrary and $y$ in those is fresh.

The unfolding rule for left inductive predicates is 
\[
\infer[(\ul)]{\Gamma,P(\vec{u})\vdash\Delta}{\mbox{all case distinctions for }P(\vec{u})},
\]
where the case distinction of $\Gamma,P(\vec{u})\vdash\Delta$ for a production
\[
\infer[]{P(\vec{t}(\vec{x}))}{Q_1(\vec{u_1}(\vec{x}))&\cdots& Q_n(\vec{u_n}(\vec{x}))&P_{1}(\vec{t_1}(\vec{x}))&\ldots&P_{m}(\vec{t_m}(\vec{x}))}
\]
is 
$\Gamma,\vec{u}=\vec{t}(\vec{y}),Q_1(\vec{u_1}(\vec{y})),\ldots, Q_n(\vec{u_n}(\vec{y})), P_{1}(\vec{t_1}(\vec{y})),\ldots, P_{m}(\vec{t_m}(\vec{y}))\vdash\Delta$, where $\vec{y}$ is fresh.

The unfolding rule for right inductive predicates is 
{\small
\[
\infer[(\ur)]{\Gamma\vdash P(\vec{t}(\vec{u})),\Delta}
    {\Gamma\vdash Q_1(\vec{u_1}(\vec{u})),\Delta
    &\cdots
    &\Gamma\vdash Q_n(\vec{u_n}(\vec{u})),\Delta
    &\Gamma\vdash P_{1}(\vec{t_1}(\vec{u})),\Delta
    &\cdots
    &\Gamma\vdash P_{m}(\vec{t_m}(\vec{u})),\Delta}
\]
}
for a production 
\[
\infer[]{P(\vec{t}(\vec{x}))}{Q_1(\vec{u_1}(\vec{x}))&\cdots& Q_n(\vec{u_n}(\vec{x}))&P_{1}(\vec{t_1}(\vec{x}))&\ldots&P_{m}(\vec{t_m}(\vec{x}))}.
\]
\end{definition}

A pre-proof of $\LKomega$ is a finite or infinite derivation tree, whose all leaves are axioms, defined in a usual way by the inference rules of $\LKomega$. A \emph{pre-proof} is denoted by $\pr$. 
The set of sequent occurrences in $\pr$ is denoted by $\Ent(\pr)$.
The sets of free variables in a formula $\phi$, a sequent $e$, and a (pre-)proof $\pr$ are denoted by $\FV(\phi)$, $\FV(e)$, and $\FV(\pr)$, respectively.
The set of substitutions of $(\subst)$ in $\pr$ is denoted by $\prsubst(\pr)$.
The number of elements in a finite set $S$ is denoted by $\# S$.

To define the proof of $\LKomega$, we define traces following paths.
\begin{definition}[Paths]
Let $\pr$ be a pre-proof in $\LKomega$.
Viewing $\pr$ as a (possibly infinite) graph, spanning from root to leaf.
A path in this graph is called a \emph{path} of $\pr$, denoted by $(e_i)_{j\leq i< k}$ for $j\in \mathbb{N}$ and $k\in\mathbb{N}\cup\{\infty\}$.
\end{definition}

\begin{definition}[Traces] Let $(e_i)_{j\leq i< k}$ be a path in a derivation tree $\pr$.
A \emph{trace} following $(e_i)_{j\leq i< k}$, denoted by $(C_i)_{j\leq i< k}$, is a sequence of inductive predicates in the antecedents of $e_i$
satisfying the following conditions:\\
(i) If $e_i$ is a conclusion of $(\ul)$ and $C_i$ is unfolded by $(\ul)$, then $C_{i+1}$ occurs in the result of the unfolding,\\
(ii) Otherwise, $C_{i+1}$ is the inductive predicate occurrence in $e_{i+1}$ corresponding to $C_i$ in $e_i$.\\
In the case (ii), $i$ is a progressing point of $(C_i)_{j\leq i< k}$. A trace with infinitely many progressing points is called an \emph{infinitely progressing trace}.
\end{definition}

We define proofs of $\LKomega$.
\begin{definition}[Proofs of $\LKomega$]
The condition for the pre-proof that, for any infinite path $(e_i)_{i\geq 0}$, there is an infinitely progressing trace following $(e_i)_{j\leq i}$ for some $j$ is called the \emph{global trace condition}.
A pre-proof $\pr$ is a proof of $\LKomega$ if $\pr$ satisfies the global trace condition.
\end{definition}

\begin{example}
\label{ex: proof NEO inf}
A proof shown in Fig.~\ref{fig: N E O inf} is a $\LKomega$ proof of $\Np(x)\vdash\Ep(x)\vee\Op(x)$. 
This proof is constructed by infinitely unfolding the proof shown in Fig.~\ref{fig: N E O}.
This proof satisfies the global trace condition.
Following the infinite path in this proof, the $\Np$'s in the antecedents construct an infinitely progressing trace.
\end{example}

\begin{figure}
    \centering
    \includegraphics[width=\linewidth]{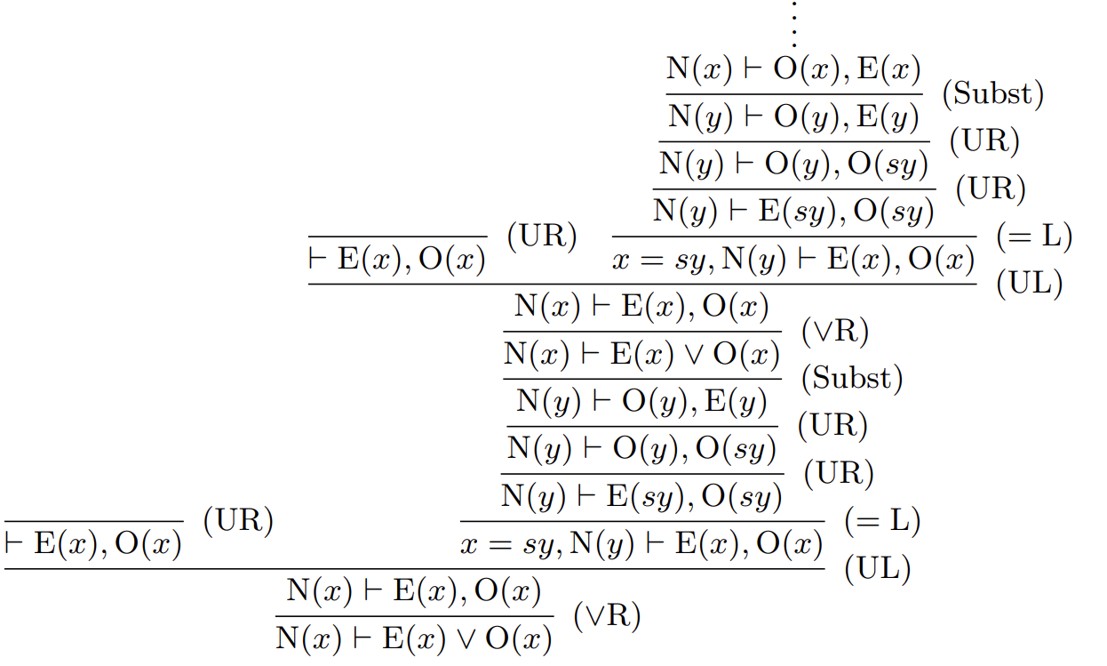}
    \caption{An infinite proof of $\Np(x)\vdash\Ep(x)\vee\Op(x)$}
    \label{fig: N E O inf}
\end{figure}

The soundness of $\LKomega$ can be proved in a similar way to previous studies\cite{Brotherston06}.
\begin{theorem}[Soundness of $\LKomega$] 
If $\pr$ is a proof of $\LKomega$, then every sequent in $\pr$ is valid.
\end{theorem}

\subsection{Cyclic-Proof System $\CLK$}
In this subsection, we define the cyclic-proof system $\CLK$ in the standard way.
The inference rules of $\CLK$ are the same as $\LKomega$.
Regarding pre-proofs, since they differ from those in $\LKomega$, we define them as follows.

\begin{definition}[Pre-proofs of $\CLK$]
A pre-proof of $\CLK$ is composed of a finite derivation tree $D$ together with a relation $R$ between sequents in $D$. It is not required that all leaves of $D$ are axioms. Leaves of $D$ that are not axioms are called \emph{buds}, and the domain of $R$ is the set of buds of $D$. The codomain of $R$ is the subset of internal nodes of $D$, and $R$ associates each bud with an internal node labeled by a syntactically identical sequent.
This sequent associated with each bud is called its \emph{companion}.
\end{definition}

A path of $\CLK$ is defined as follows.
\begin{definition}[Paths of $\CLK$] Let $\pr=(D,R)$ be a pre-proof in $\CLK$.
Viewing $D$ as a directed graph, we add edges according to $R$.
A (finite or infinite) path in this graph is called a path of $\pr$.
\end{definition}

Following the above definition of paths, a trace, the global trace condition, and a proof of $\CLK$ are defined analogously to those in $\LKomega$. The soundness of $\CLK$ can be proved by the soundness of $\LKomega$ since whenever there exists a proof in $\CLK$, there also exists a proof containing the same sequent in $\LKomega$.

\begin{example}
\label{ex: proof NEO cycle}
The proof shown in Fig.~\ref{fig: N E O} is a proof of $\CLK$. 
This proof satisfies the global trace condition.
Following the infinite path in this proof, the underlined $\Np$'s on the antecedents construct an infinitely progressing trace.
\end{example}
\section{Admissibility of the Substitution Rule in $\CLK$}

In this section, we show the admissibility of $\CLK$. 
To prove this, we assume a proof $\prplus$ of $\Gamma\vdash\Delta$ in $\CLK$ with $(\subst)$ and construct a proof $\prminus$ of $\Gamma\vdash\Delta$ without $(\subst)$. 

The construction procedure can be divided into the elimination of composite substitutions and atomic substitutions.
Concerning the elimination of composite substitutions, from a proof $\prplus$ containing composite substitutions, we construct a proof $\prvar$ containing only atomic substitutions by a proof transformation.

Concerning the elimination of atomic substitutions, it can be outlined in three steps as shown in Fig.~\ref{fig: outline}.

\begin{figure}
    \centering
    \includegraphics[width=0.8\linewidth]{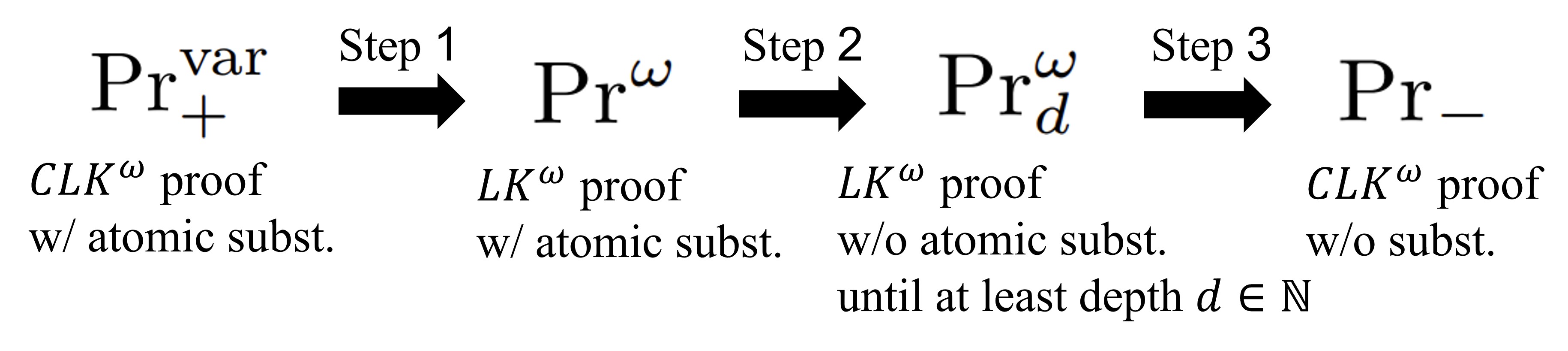}
    \caption{The outline of the elimination of atomic substitutions}
    \label{fig: outline}
\end{figure}

First, we construct a $\LKomega$ proof $\promega$ by unfolding the cycles in $\prvar$.
Secondly, we perform a lifting of the substitution rule to $\promega$. 
Taking $\promega$ as the base case $\prd{0}$, we recursively construct, for each depth $d$, $\LKomega$ proofs $\prd{d}$ that contain no substitution rules until at least depth $d$.
Finally, from $\prd{d}$ for a sufficiently large $d$, we obtain a $\CLK$ pre-proof $\prminus$ without substitution rules.

We transform the proof stepwise in the elimination of atomic substitutions, but the resulting pre-proof $\prminus$ must satisfy the global trace condition.
To achieve this, at each stage, we define a mapping that associates occurrences of sequents in the constructed proof to those in $\prvar$.

In the third step, we need to choose buds and their companions among the sequent occurrences in $\prd{d}$.
At this stage, we choose pairs of sequent occurrences so that each bud and its companion correspond to the same sequent occurrence in $\prvar$. 
As a result, every path in $\prminus$ corresponds to a path in $\prvar$, and
any trace following the path in $\prminus$ is associated with a trace following the corresponding path in $\prvar$. $\prplus$ satisfies the global trace condition, and so does $\prvar$.

\subsection{The Elimination of Composite Substitutions}
In this subsection, we construct $\prvar$ with only atomic substitutions from $\prplus$.
\begin{lemma}[From $\prplus$ to $\prvar$]
\label{lemm: prplus to prvar}
Let $\prplus$ be a $\CLK$ proof of $\Gamma\vdash\Delta$. 
Then, there exist a $\CLK$ proof $\prvar$ of $\Gamma\vdash\Delta$ such that any substitution in $\prsubst(\prvar)$ is atomic.
\end{lemma}

\begin{proof}
For the instance of $(\subst)$
\[
\infer[(\subst)]{\Gamma[x_1:=t_1,\ldots,x_n:=t_n]\vdash\Delta[x_1:=t_1,\ldots,x_n:=t_n]}{\Gamma\vdash\Delta}
\]
for a composite substitution, we can transform it to\\
\begin{adjustbox}{max width=1.2\linewidth,center}
$
\infer[(\cut)]{\Gamma[x_1:=t_1,\ldots,x_n:=t_n]\vdash\Delta[x_1:=t_1,\ldots,x_n:=t_n]}
    {\infer[(\existsR)]{\vdash\exists y_1.y_1=t_1}{\infer[(=R)]{t_1=t_1}{}}
    &\infer*[(\cut)\mbox{ as well as the bottom}]{\exists y_1.y_1=t_1,\Gamma[x_1:=t_1,\ldots,x_n:=t_n]\vdash\Delta[x_1:=t_1,\ldots,x_n:=t_n]}
        {\infer[(\existsL)\ n\ \mbox{times}]{\exists y_1.y_1=t_1,\ldots,\exists y_n=t_n,\Gamma[x_1:=t_1,\ldots,x_n:=t_n]\vdash\Delta[x_1:=t_1,\ldots,x_n:=t_n]}
            {
                \infer[(\eqL)\ n\ \mbox{times}]{y_1=t_1,\ldots, y_n=t_n,\Gamma[x_1:=t_1,\ldots,x_n:=t_n]\vdash\Delta[x_1:=t_1,\ldots,x_n:=t_n]}
                    {
                        \infer[(\wk)]{y_1=t_1,\ldots, y_n=t_n,\Gamma[x_1:=y_1,\ldots,x_n:=y_n]\vdash\Delta[x_1:=y_1,\ldots,x_n:=y_n]}
                            {\infer[(\subst)]{\Gamma[x_1:=y_1,\ldots,x_n:=y_n]\vdash\Delta[x_1:=y_1,\ldots,x_n:=y_n]}{\Gamma\vdash\Delta}}
                    }
            }
        }
    },
$
\end{adjustbox}
where $y_i(i\in[1,n])$ are fresh variables for $\Gamma\vdash\Delta$ and $(\Gamma\vdash\Delta)[x_1:=t_1,\ldots,x_n:=t_n]$.  
Since every sequent occurring in the original proof also appears after the transformation,
and the transformation preserves every trace, the transformed proof satisfies the global trace condition.  \qed
\end{proof}

\begin{remark}
We construct $\prvar$ by a proof transformation that employs several rules.
Therefore, this lemma depends on the inference rules of $\CLK$.
In particular, the above transformation is not possible in cut-free
 $\CLK$, which is strictly weaker than $\CLK$ with
 cuts~\cite{Masuoka23}.
\end{remark}

\subsection{Admissibility of the Atomic Substitution Rule}
In this subsection, we prove the admissibility of the atomic substitution rule in $\CLK$.
Note that in the proof transformation of the previous subsection, constants also disappear; however, in the present proof, we do not make this assumption. 

Hereafter, suppose that $\prvar$ is a $\CLK$ proof such that $\prsubst(\prvar)$ is atomic.
First, we construct a $\LKomega$ proof $\promega$ from $\prvar$. 
At this stage, we define a mapping $\fomega$ from sequent occurrences in $\promega$ to those in $\prvar$.

\begin{lemma}[From $\prvar$ to $\promega$] 
\label{lemm: prplus to promega}
Let $\Gamma\vdash\Delta$ be the conclusion of $\prvar$.
There exists a $\LKomega$ proof $\promega$ of $\Gamma\vdash\Delta$ s.t. $\FV(\promega) = \FV(\prvar)$ and $\prsubst(\promega) = \prsubst(\prvar)$ hold. 
Moreover, there exists a function $\fomega: \Ent(\promega)\to\Ent(\prvar)$ s.t. the following conditions hold:
\begin{itemize}
    \item $\fomega(e_0)= e'_0$, where $e_0$ and $e_0'$ are the conclusions of $\promega$ and $\prvar$, respectively.
    \item For any $e\in\promega$, $\fomega(e)\equiv e$.
    \item If the rule inference
    \[\infer[(R)]{e_i}{e_{i+1}^1&\cdots&e_{i+1}^n}\]
    exists in $\promega$, then the rule instance
    \[\infer[(R)]{\fomega(e_i)}{\fomega(e_{i+1}^1)&\cdots&\fomega(e_{i+1}^n)}\]
    exists in $\prvar$.
    Furthermore, any traces in this rule instance of $\prvar$ are preserved in this rule instance of $\prminus$.
    Here, each bud and its companion are identified in the rule instance of $\CLK$, and $\fomega(e)$ takes the companion.
\end{itemize}
\end{lemma}

$\promega$ can be constructed in a similar manner to the tree unfolding
as defined by Brotherston~\cite{Brotherston06}.  In the tree unfolding,
cycles are unfolded based on the relation between buds and companions.
Fig.~\ref{fig: N E O inf} is an infinite proof obtained by unfolding
Fig.~\ref{fig: N E O} in this way.  Intuitively, $\fomega$ associates
each sequent in the infinite proof obtained by unfolding the cyclic
proof with the corresponding sequent in the original cyclic proof.

The preservation of traces means that, for any trace following the path from $\fomega(e_i)$ to $\fomega(e^m_{i+1})$, there exists a corresponding trace following the path from $e_i$ to $e_{i+1}$ by the corresponding inductive predicates.


Secondly, we mention the recursive construction of $\LKomega$ proofs $\prd{d}$ of $\Gamma\vdash\Delta$ that contain no substitution rules until at least depth $d$. 
At this stage, we define $\fd{d}:\Ent(\prd{d})\to\Ent(\prvar)\cup\{\epsilon\}$. $\epsilon$ is a special symbol indicating no correspondence.
In correspondence with the lifting of substitution rules, we introduce \emph{partial-substitution closure} and \emph{substitution-application property} for atomic substitutions.

\begin{definition}[Partial-substitution closure]
Let $\Theta$ and $X$ be a finite set of substitutions and a finite set of variables, respectively.
We inductively define the \emph{partial-substitution closure} $\psclosure(\Theta, X)$ as the smallest set which satisfies the following:
\begin{itemize}
    \item $\Theta\subseteq\psclosure(\Theta, X)$,
    \item $\forall\theta\in\psclosure(\Theta, X),x,y\in X.\theta[x\to y]\in\psclosure(\Theta,X)$,
    \item $\forall\theta_1,\theta_2\in\psclosure(\Theta, X).\theta_1\theta_2\in\psclosure(\Theta,X)$.
\end{itemize}
\end{definition}

The partial-substitution closure $\psclosure(\Theta, X)$ is the smallest set containing $\Theta$ closed under two operations: 
overwriting a substitution with variables in $X$, and the composition of substitutions. 
When substitutions are restricted to atomic ones, the domain and image of the partial-substitution closure are fixed as finite sets, allowing the following lemma to be established.

\begin{lemma}
\label{lemm: psc atomic and finite}
Let $\Theta_{var}$ be a finite set of atomic substitutions.
Then, any substitution in partial-substitution closure $\psclosure(\Theta_{var}, X)$ is an atomic substitution for any variable set $X$.
Furthermore, $\psclosure(\Theta_{var}, X)$ is finite.
\end{lemma}

\begin{remark}
When $\Theta_{var}$ contains a composite substitution, $\psclosure(\Theta_{var}, X)$ can be infinite.
For example, suppose there is a unary function symbol $f$ and some $\theta\in \Theta_{var}$ such that $f(x)$ occurs in the image of $\theta$. Then the terms $f(x), f(f(x)),\ldots$ are all appear in images of substitutions in $\psclosure(\Theta_{var}, X)$, which force $\psclosure(\Theta_{var}, X)$ to infinite.
\end{remark}

\begin{definition}[Substitution-application property]
 Consider a rule $(R)$ except $(\subst)$. If for any instance
\[
\infer[(R)]{\Gamma\vdash\Delta}{\Gamma_1\vdash\Delta_1 & \ldots & \Gamma_n\vdash \Delta_n}
\]
of $(R)$ and any atomic substitution $\theta$, there is an instance of $(R)$:
\[
\infer[(R)]{\Gamma[\theta] \vdash\Delta[\theta]}{\Gamma_1[\theta_1]\vdash\Delta_1[\theta_1] & \ldots & \Gamma_n[\theta_n]\vdash\Delta_n[\theta_n]},
\]
where $\theta_i(1\leq i\leq n)$ are substitutions in $\psclosure(\{\theta\}, \FV(\Gamma_i\vdash \Delta_i)\cup\FV(\Gamma\vdash\Delta))$,
s.t. any traces in the above instance are preserved, then we say that $(R)$ satisfies the \emph{substitution-application property}.
We refer to the latter instance as a \emph{substitution application} of the former.
If all inference rules of a proof system $\system$ satisfy the substitution-application property, 
then we say that $\system$ satisfies the substitution-application property.
\end{definition}

Note that the substitution-application property subsumes closure under atomic substitution for axiom rules.
Below, we discuss the fact that $\CLK$ and $\LKomega$ satisfy the substitution-application property.
\begin{lemma}
\label{lemm: subst-app property}
The proof systems $\CLK$ and $\LKomega$ satisfy the substitution-application property.
\end{lemma}
\begin{proof}
If the rule $(R)$ does not introduce any fresh variables, it suffices to take $\theta_i=\theta$ for $0\leq i\leq n$. 
Hereafter, we only consider the case of $(\ul)$.
The cases where fresh variables appear in other rules are handled in the same manner.

Consider an instance
\[
\infer[(\ul)]{\Gamma\vdash\Delta}{\Gamma_1\vdash\Delta & \ldots &\Gamma_n\vdash\Delta_n}
\]
of $(\ul)$ and an atomic substitution $\theta$.

Let $f|_D$ denote the partial function of $f$ restricted to the domain $D$, let $Var(T)$ denote the set of variables occurring in the terms of $T$, and let $X_i$ be $\FV(\Gamma\vdash\Delta)\cup\FV(\Gamma_i\vdash\Delta)$ for each $1\leq i\leq n$. 

For each $i$, we can define $\theta_i\in\psclosure(\{\theta\},X_i)$ satisfying
\begin{itemize}
    \item $\theta_i|_{\FV(\Gamma\vdash\Delta)}$ equals to $\theta|_{\FV(\Gamma\vdash\Delta)}$,
    \item any variable in the image of $\theta_i|_{X_i-\FV(\Gamma\vdash\Delta)}$ is fresh for $\Gamma[\theta]\vdash\Delta[\theta]$, and
    \item $\theta_i|_{X_i-\FV(\Gamma\vdash\Delta)}$ is injective.
\end{itemize}
Note that $X_i-\FV(\Gamma\vdash\Delta)$ is the set of variables in
 $\Gamma_i\vdash\Delta$ that are fresh for $\Gamma\vdash\Delta$.
 Let $\{x_1,\ldots,x_a\}=X_i-\FV(\Gamma\vdash\Delta)$.  Since
 $\theta$ is atomic, we have $\#\FV(\Gamma\vdash\Delta)\geq
 \#Var(\img(\theta_i|_{\FV(\Gamma\vdash\Delta)}))$, and we can
 choose pairwise distinct $x'_{j}\in
 X_i-Var(\img(\theta_i|_{\FV(\Gamma\vdash\Delta)}))$ for each $1\leq
 j\leq a$. Hence, we can define
\[
 \theta_i = \theta[x_{1}\to x'_{1},\ldots, x_{a}\to x'_{a}]
 \in\psclosure(\{\theta\},X_i)
\]
 such that the above conditions hold, and then
\[
\infer[(\ul)]{\Gamma[\theta]\vdash\Delta[\theta]}{\Gamma_1[\theta_1]\vdash\Delta[\theta_1] & \ldots & \Gamma_n[\theta_n]\vdash\Delta[\theta_n]},
\]
is an instance of $(\ul)$.
\qed
\end{proof}

\begin{remark}
In general, the substitution-application property does not hold for composite substitutions.
As an example, consider a rule instance 
\[
\infer[(\ul)]{\Op(x)\vdash \Np(x)}{x=sy,\Ep(y)\vdash \Np(x)}
\]
and a substitution $\theta=[f(x,y)/x]$.
The free variables contained in this rule instance are exactly $x$ and $y$.
On the other hand, when applying $(\ul)$ to $\Op(f(x,y))\vdash \Np(f(x,y))$, a fresh variable different from $x$ and $y$ is required in the premise.
Such a variable does not appear in the image of substitution of $\psclosure(\{[f(x,y)/x]\},\{x,y\})$.
\end{remark}

By the substitution application, $\prd{d}$ can be constructed from $\promega$.
\begin{lemma}[From $\promega$ to $\prd{d}$] 
\label{lemm: promega to prd}
Let $\Gamma\vdash\Delta$ be the conclusion of $\prvar$.
There exists a $\LKomega$ proof $\prd{d}$ of $\Gamma\vdash\Delta$ that contains no substitution rules until at least depth $d$ s.t. $\FV(\prd{d})\subseteq \FV(\prvar)$ and $\Theta(\prd{d})\subseteq \psclosure(\Theta(\prvar), \FV(\prvar))$ hold. 
Moreover, there exists a function $\fd{d} : \Ent(\prd{d})\to\Ent(\prvar)\cup \{\epsilon\}$ s.t. the following conditions hold:
\begin{itemize}
    \item $\fd{d}(e_0)= e'_0$, where $e_0$ and $e_0'$ are the conclusions of $\prd{d}$ and $\prvar$, respectively.
    \item $\fd{d}(e)=\epsilon$ if $e$ is a premise of the $(\subst)$.
    \item If $\fd{d}(e)\neq\epsilon$, then $e\equiv \fd{d}(e)\theta$ for some $\theta\in\psclosure(\prsubst(\prvar),\FV(\prvar))$.
    \item If the following applications of $(\subst)$ and some inference rule $(R)$ except $(\subst)$ exist in $\prd{d}$:
    \[
    \infer*[(\subst)\mbox{ 0 or more times}]{\infer[]{e_i}{\qquad\qquad}}{\infer[(R)]{}{e_{i+1}^1&\cdots&e_{i+1}^n}},
    \]
    then the following applications of  $(\subst)$ and $(R)$ exist in $\prvar$:
    \[
    \infer*[(\subst)\mbox{ 0 or more times}]{\infer[]{\fd{d}(e_i)}{\qquad\qquad}}{\infer[(R)]{}{\fd{d}(e_{i+1}^1)&\cdots&\fd{d}(e_{i+1}^n)}}.
    \]
    Furthermore, any traces in these rule applications of $\prvar$ are preserved in $\prd{d}$.
    Here, each bud and its companion are identified in the rule applications of $\CLK$, and $\fd{d}(e)$ takes the companion.
\end{itemize}
\end{lemma}
\begin{proof}
The proof proceeds by induction on $d$.

{\sffamily Base Case:} $\prd{0}$ is $\promega$ in Lemma~\ref{lemm: prplus to promega}. 
$\fd{0}$ is $\fomega$ in Lemma~\ref{lemm: prplus to promega}, except for the premises of $(\subst)$ in $\prd{0}$. For the premises of $(\subst)$, we assign $\fd{0}(e)=\epsilon$.

{\sffamily Induction Step:} Assume $\prd{k}$ and $\fd{k}$. 
By performing the following operation on every $(\subst)$ at depth $k+1$ in $\prd{k}$, we can construct $\prd{k+1}$.
At the same time, we define $\fd{k+1}$. 
In parts of the proof unaffected by this transformation, $\fd{k+1}$ keeps the values of $\fd{k}$.

If the premise of the $(\subst)$ is derived from zero or more applications of $(\subst)$ and one application of some rule $(R)$, that part is of the following form.
\[
\infer[]{\Gamma[\theta_1\cdots\theta_n]\vdash \Delta[\theta_1\cdots\theta_n]}
    {\infer*[(\subst)\mbox{ }n-1\mbox{ times}]{}{
        \infer[(\subst)]{\Gamma[\theta_1]\vdash \Delta[\theta_1]}
            {\infer[(R)]{\Gamma\vdash \Delta}
                {\infer*{\Gamma_1\vdash \Delta_1}{} 
                &\ldots 
                &\infer*{\Gamma_m\vdash \Delta_m}{}
                }
            }
        }
    }.
\]
At this point, since $\theta_1\cdots\theta_n$ is an atomic substitution, we apply the substitution application of $(R)$ to $\Gamma[\theta_1\cdots\theta_n]\vdash \Delta[\theta_1\cdots\theta_n]$ in $\prd{k+1}$ as follows:
\[
\infer[(R)]{
\Gamma[\theta_1\cdots\theta_n]\vdash \Delta[\theta_1\cdots\theta_n]
}{
\infer[(\subst)]{\Gamma_1[\theta'_1]\vdash \Delta_1[\theta'_1]}
    {\infer*{\Gamma_1\vdash \Delta_1}{}}
&
\ldots
&
\infer[(\subst)]{\Gamma_m[\theta'_m]\vdash\Delta_m[\theta'_m]}
    {\infer*{\Gamma_m\vdash\Delta_m}{}}
},
\]
where every $\theta'_i$ is in $\psclosure$$(\{\theta_1\cdots\theta_n\}, \FV(\Gamma_i\vdash \Delta_i)\cup \FV(\Gamma\vdash \Delta))\subseteq\psclosure(\prsubst(\prvar),\FV(\prvar))$ by Lemma~\ref{lemm: subst-app property}.
$\fd{k+1}(\Gamma_i\vdash \Delta_i)$ is defined as $\epsilon$. 
$\fd{k+1}(\Gamma_i[\theta'_i]\vdash \Delta_i[\theta'_i])$ is defined as $\fd{k}(\Gamma_i\vdash \Delta_i)$. 
$\fd{k+1}(\Gamma[\theta_1\cdots\theta_n]\vdash \Delta[\theta_1\cdots\theta_n])$ is defined as $\fd{k}(\Gamma[\theta_1\cdots\theta_n]\vdash \Delta[\theta_1\cdots\theta_n])$.
As a result, the $(\subst)$ applications at depth $k+1$ are pushed up to $k+2$.

From the above, we can construct $\prd{k+1}$ with $\prd{k}$. 
Since the above operation affects only finite portions of each path, $\prd{d}$ satisfies the global trace condition.
Moreover, since the proof transformation is performed with substitution applications, the conditions $\FV(\prd{d})\subseteq \FV(\prplus)$ and $\Theta(\prd{d})\subseteq \psclosure(\Theta(\prplus), \FV(\prplus))$ are preserved in every induction step.

In defining $\fd{k+1}$, we associate sequent occurrences in $\promega$ with corresponding sequent occurrences in $\prvar$ by extending $\fd{k}$.
Note that these substitution applications preserve the conditions specified in this lemma, including the preservation of traces.\qed
\end{proof}

In the third place, we demonstrate that we can construct the pre-proof $\prminus$ in $\CLK$ from $\prd{d}$ for sufficiently large $d$. 
Intuitively, we show that for each infinite path in $\prd{d}$, there are sequents $e_i$ and $e_j$ s.t. $i< j\leq d$, $\fd{d}(e_i)=\fd{d}(e_j)$, and $e_i\equiv e_j$ hold. By assigning such $e_j$ and $e_i$ as a bud and its companion, respectively, we can construct $\prminus$.

\begin{lemma}[From $\prd{d}$ to $\prminus$]
\label{lemm: prd to prminus}
Let $\Gamma\vdash\Delta$ be the conclusion of $\prvar$.
There exists a $\CLK$ pre-proof $\prminus$ of $\Gamma\vdash\Delta$ that contains no substitution rules, 
and a function $\f : \Ent(\prminus)\to\Ent(\prvar)$ s.t. the following conditions hold:
\begin{itemize}
    \item $\f(e_0)= e'_0$, where $e_0$ and $e_0'$ are the conclusions of $\prminus$ and $\prvar$, respectively.
    \item For any bud $e_i$ and its companion $e_{i+1}$, $\f(e_i)= \f(e_{i+1})$.
    \item If the rule instance
    \[\infer[(R)]{e_i}{e_{i+1}^1&\cdots&e_{i+1}^n}\] 
    exists in $\prminus$, then the following applications of $(\subst)$ and $(R)$ exist in $\prvar$:
    \[
    \infer*[(\subst)\mbox{ 0 or more times}]{\infer[]{\fd{d}(e_i)}{\qquad\qquad}}{\infer[(R)]{}{\fd{d}(e_{i+1}^1)&\cdots&\fd{d}(e_{i+1}^n)}}.
    \]
    Furthermore, any traces in these rule applications of $\prvar$ are preserved in this rule application of $\prminus$.
    Here, each bud and its companion are identified in the rule applications, and $\f(e)$ takes the companion.
\end{itemize}    
\end{lemma}
\begin{proof}
Let $\prd{d}$ and $\fd{d}$ be a proof and a mapping constructed by Lemma~\ref{lemm: promega to prd}.
Consider $n=\#\psclosure(\prsubst(\prvar),\FV(\prvar))$ by Lemma~\ref{lemm: psc atomic and finite}.

Assume that $d> \#\img(\fd{d})\cdot n$ holds.
Such $d$ exists since $\img(\fd{d})\subseteq \Ent(\prvar)\cup\{\epsilon\}$ holds.
For any path $(e_i)_{0\leq i\leq d}$ in $\prd{d}$, for some sequent $e$ in $\prvar$, there exist at least $n+1$ indices $i$ with $0\leq i\leq d$ such that $\fd{d}(e_i)=e$.
Let $S$ denote the set of those indices.
$\fd{d}(e_i)\neq\epsilon$ holds for any $i\in S$ since $e_i$ is not a premise of $(\subst)$. 
Therefore, $e_i$ and $\fd{d}(e_i)\theta_i$ are the same sequents for some $\theta_i\in\psclosure(\prsubst(\prvar),\FV(\prvar))$ by the condition of $\fd{d}$.
The number of substitutions in $\psclosure(\prsubst(\prvar),\FV(\prvar))$ is $n$. 
Hence, there are the same sequents $e_i$ and $e_j$ for some  $i<j\in S$. 
By choosing $e_i$ as the bud and $e_j$ as its companion, we can construct $\prminus$.
Moreover, for each occurrence of a sequent $e$ in $\prminus$, there exists a corresponding occurrence of a sequent $e'$ in $\prd{d}$. 
Based on this correspondence, we define $\f(e)=\fd{d}(e')$. 
With the definition of $\f$, all conditions required by this lemma are satisfied.\qed
\end{proof}

From the above, we could construct the pre-proof $\prminus$ from a proof $\prvar$ via $\promega$ and $\prd{d}$.
Next, we mention that the pre-proof $\prminus$ is a proof. 
In other words, we show that the pre-proof $\prminus$ satisfies the global trace condition. 
First, define a \emph{corresponding path} in $\prvar$ for each path in $\prminus$.

\begin{definition}[Corresponding path]
Let $\prminus$ and $\f$ be a pre-proof and a mapping, respectively, constructed in Lemma~\ref{lemm: prd to prminus}.
Let $(e_i)_{0\leq i}$ be a path in $\prminus$ where $e_0$ is the conclusion of $\prminus$.
The path $(e'_j)_{0\leq j}$ in $\prvar$ defined as follows is called \emph{corresponding path} of $(e_i)_{0\leq i}$:
\begin{itemize}
    \item Let $e'_0= \f(e_0)$ be the conclusion of $\prvar$.
    \item Suppose that $e'_j=\f(e_i)$. 
    \begin{itemize}
        \item If $e_i$ is a bud and $e_{i+1}$ is its companion, then let $e'_j=\f(e_{i+1})$.
        \item If $e_{i+1}$ is the $n$-th premise in a rule instance
        \[\infer[(R)]{e_i}{\cdots&e_{i+1}&\cdots},\]
        then, by the condition of $\f$, there exists rule applications 
        \[
            \infer*[(\subst)\mbox{ and bud-companion $m$ times}]
                {\infer[]{e'_j=\fd{d}(e_i)}{\qquad\qquad}}
                    {\infer[(R)]{}{\cdots&\fd{d}(e_{i+1})&\cdots}},
        \]
        in $\prvar$ for some $m$ such that $\fd{d}(e_{i+1})$ is the $n$-th premise of $(R)$. In this application, the sequence $e'_j, \ldots{}, e'_{j+m}$ is uniquely determined since $(\subst)$ is a unary rule. Furthermore, we set $e'_{j+m+1}= f(e_{i+1})$.
    \end{itemize}
\end{itemize}    
\end{definition}

Using the corresponding path, we can refer to the relationship between the traces of $\prvar$ and $\prminus$.
\begin{lemma}
\label{lemm: trace from plus to minus}
Let $(e_i)_{0\leq i}$ be a path in $\prminus$ constructed in Lemma~\ref{lemm: prd to prminus} and $(e'_j)_{0\leq j}$ be the corresponding path of $(e_i)_{0\leq i}$ in $\prvar$.
If there is an infinitely progressing trace $(C'_j)_{k'\leq j}$ following $(e'_j)_{0\leq j}$, the infinitely progressing trace
 $(C_i)_{k\leq i}$ exists in $(e_i)_{0\leq i}$. 
\end{lemma}

Based on Lemma~\ref{lemm: prplus to prvar} and~\ref{lemm: prd to prminus}, we can construct a pre-proof without the substitution rule. Lemma~\ref{lemm: trace from plus to minus} guarantees that the pre-proof satisfies the global trace condition; hence, the following main theorem holds.
\begin{theorem}[Admissibility of substitution rules in $\CLK$]
\label{thm: adm subst in CSL}
If $\Gamma\vdash\Delta$ is provable in $\CLK$, then it is provable in $\CLK$ without $(\subst)$.
\end{theorem}

In the process of elimination of atomic substitutions, we do not introduce any rule not present in $\prvar$, and hence we have the following stronger result.
\begin{theorem}
If $\pr$ is a proof of $\Gamma\vdash\Delta$ in $\CLK$ and $\prsubst(\pr)$ consists of only atomic substitutions, then there is a $(\subst)$-free proof of $\Gamma\vdash\Delta$ that contains only rules in $\pr$.
\end{theorem}
\label{sec: admissibility subst in LK}
\section{Admissibility of the Substitution Rule in Other Cyclic-Proof Systems}
\label{sec: admissibility subst in others}

In this section, we discuss the conditions under which our procedure for the substitution elimination can be applied in general cyclic-proof systems.
Among the lemmas proved in Section~\ref{sec: admissibility subst in LK}, those that depend on the proof system $\CLK$ are Lemma~\ref{lemm: prplus to prvar}~and~\ref{lemm: subst-app property}.

First, we discuss Lemma~\ref{lemm: prplus to prvar}, which concerns the elimination of composite substitutions.
This lemma claims the existence of a proof that contains only atomic substitutions from the existence of a proof that allows composite substitutions. 
To prove this lemma, we use inference rules $(\cut)$, $(\eqL)$, $(\eqR)$, $(\existsL)$, and $(\existsR)$.
Accordingly, in proof systems that do not include these rules, this lemma does not, in general, hold.

In fact, we can consider a system in which $(\subst)$ is not admissible.
For example, let $P(x)$ be an inductive predicate defined by $P(x):= P(sx)$. The unfolding rule of the left $P(x)$ in $\CLK$ is 
\[
\infer[(\ul)]{P(x),\Gamma\vdash\Delta}{y=x,P(sy),\Gamma\vdash\Delta}.
\]

However, we can consider another unfolding rule
\[
\infer[(\ul')]{P(x),\Gamma\vdash\Delta}{P(sx),\Gamma\vdash\Delta}.
\]
Let we think a proof system with only $(\ul')$, $(\ur)$, $(\axiom)$, $(\wk)$, and $(\subst)$. 
There is a proof 
\begin{center}
\includegraphics[width=0.25\linewidth]{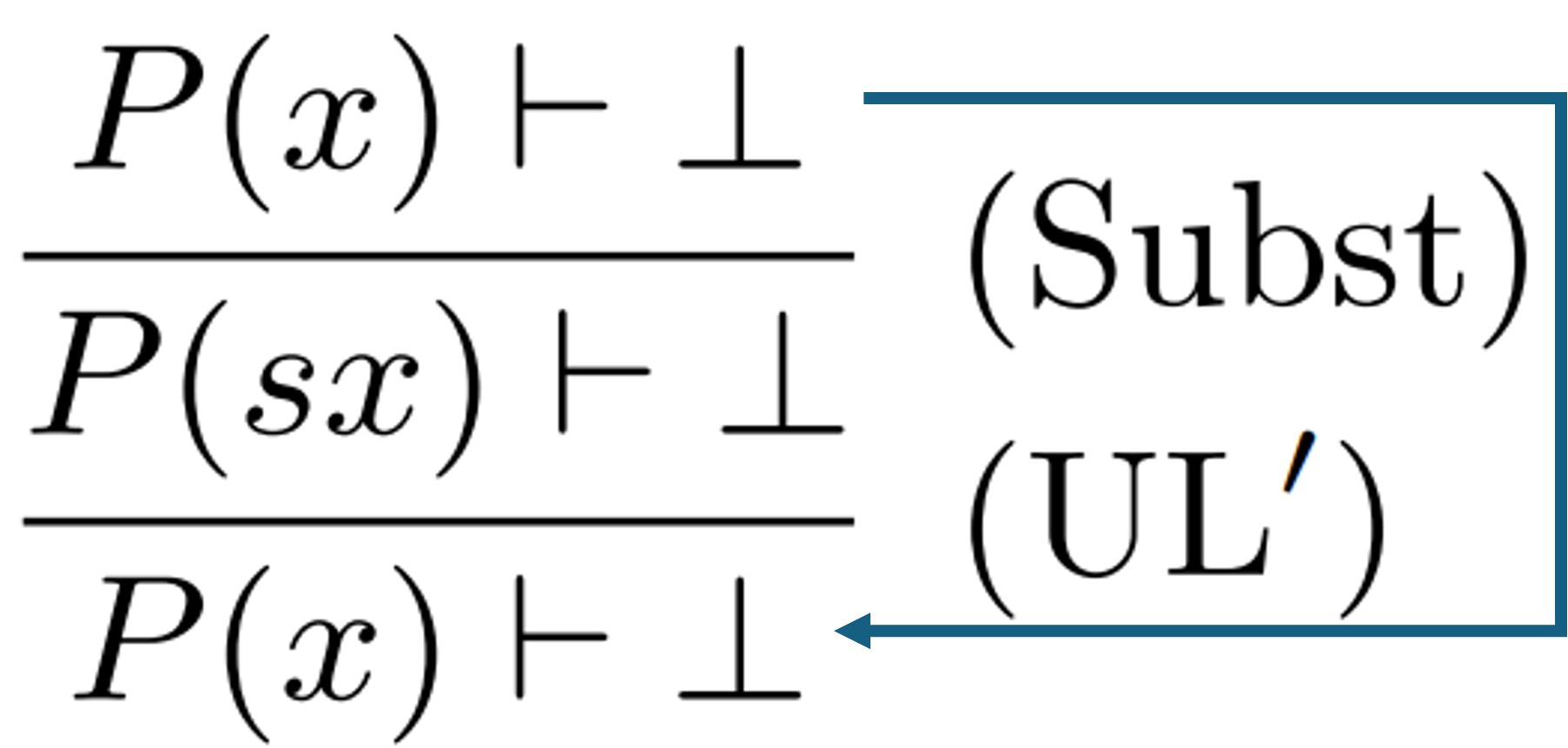}    
\end{center}
of the sequent $P(x)\vdash\bot$.
However, we cannot prove this sequent without $(\subst)$ in this system.

Note that in Lemma~\ref{lemm: prplus to prvar} for $\CLK$, the cut rule is only used to introduce $y=t$ (where $y$ is fresh and $t$ is arbitrary). Therefore, when we introduce the following rule in place of the cut rule, $(\subst)$ is still admissible
\[
\infer[(fresh~L)]{\Gamma\vdash\Delta}{y=t,\Gamma\vdash\Delta},
\]
where $t$ is arbitrary and $y$ is fresh in $\Gamma\vdash\Delta$ and $t$.

Secondly, we discuss Lemma~\ref{lemm: subst-app property}. This lemma plays an important role in the elimination of atomic substitutions.
This lemma guarantees that $\CLK$ and $\LKomega$ satisfy the substitution-application property with atomic substitutions. 
This lemma holds for many general proof systems, including the proof system $\CLK$ and the cyclic proof system for separation logic. 
Therefore, as long as we can assume a proof with only atomic substitutions, the substitution rule can be eliminated in many proof systems.
Let $\CLK_{-}$ be a proof system $\CLK$ without the cut rule.

\begin{theorem}[Admissibility of atomic substitution rules in $\CLK_{-}$]
\label{thm: adm subst in LK}
If $\pr$ is a proof of $\Gamma\vdash\Delta$ in $\CLK_{-}$ and $\prsubst(\pr)$ consists of only atomic substitutions, then there is a $(\subst)$-free $\CLK_{-}$ proof of $\Gamma\vdash\Delta$ that contains only rules in $\pr$.
\end{theorem}

In particular, every proof in the cyclic-proof system $\CSLomega$~\cite{Brotherston11b} for the separation logic contains only variable application since the separation logic contains no function symbols. 
Hence, the following theorems hold. Let $\CSLomega_{-}$ be $\CSLomega$ without the cut rule.
\begin{theorem}
\label{thm: adm subst in CSLM}
$(\subst)$ is admissible in both $\CSLomega$ and $\CSLomega_{-}$.
\end{theorem}
\section{Conclusion}
In this paper, we prove the admissibility of the substitution rule in $\CLK$.
After that, we discuss generalization to other proof systems.

One possible direction for future work is investigating the admissibility of the substitution rule in cut-free $\CLK$ without $(fresh~L)$.
Section~\ref{sec: admissibility subst in others} showed that replacing $(UL)$ with $(UL')$ and restricting inference rules destroy this admissibility.
As a counterexample, we considered $P(x)\vdash\bot$; however, in the usual $\CLK$ system, there still be a cut-free proof 
\begin{center}
\includegraphics[width=0.32\linewidth]{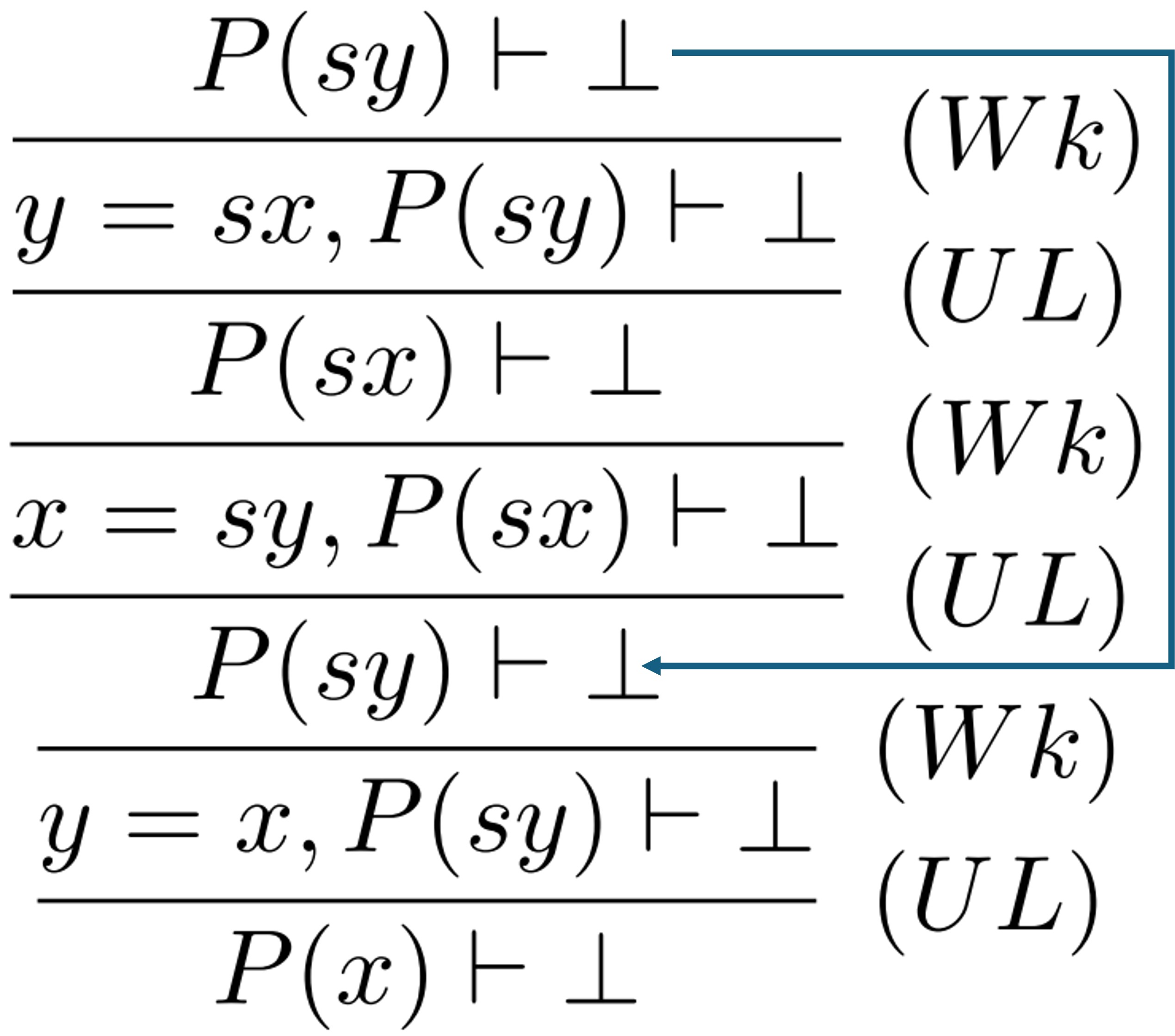}    
\end{center}
of this sequent.
The reason is that $(UL)$ effectively plays the role of $(fresh~L)$.
Our interest lies in identifying the minimal set of rules that still allows the elimination of the substitution rule in the presence of function symbols.

Another possible direction is to study the cut-elimination property for cyclic proof systems by introducing additional rules.
The fact that the substitution rule could be eliminated from cut-free $\CLK$ through the introduction of $(fresh~L)$ provides an important insight.
This suggests that, by adding rules sufficient to construct cycles, it may be possible to achieve the elimination of the cut rule.

\begin{credits}
\subsubsection{\ackname} 
The second author was supported by JSPS KAKENHI Grant Number 22K11901.
\end{credits}

%
%
%
%
\newpage
\bibliographystyle{splncs04}
\bibliography{bib}
\appendix
\section{Details of Proofs in Section~\ref{sec: admissibility subst in LK}}
In this section, we show the details of proofs in Section~\ref{sec: admissibility subst in LK}.
\subsection{Lemma~\ref{lemm: psc atomic and finite}}
\begin{proof}
Since any function symbols not contained in $\Theta$ are likewise absent from $\psclosure(\Theta, X)$, any substitution in $\psclosure(\Theta_{var}, X)$ is atomic.

Let $\dom(\Theta)$ and $\img(\Theta)$ denote, respectively, the unions of $\dom(\theta)$ and $\img(\theta)$ over all $\theta\in\Theta$. 
The domain of every substitution in $\psclosure(\Theta_{var}, X)$ is a subset of $\dom(\Theta)\cup X$.
Similary, the range of every substitution in $\psclosure(\Theta_{var}, X)$ is a subset of $\img(\Theta)\cup X$. 
Hence, the value of $\#\psclosure(\Theta_{var}, X)$ is at most $\#(\img(\Theta)\cup X)^{\#(\dom(\Theta)\cup X)}$ and $\#\psclosure(\Theta_{var}, X)$ is finite.\qed
\end{proof}

\subsection{Lemma~\ref{lemm: trace from plus to minus}}
\begin{proof}
By the property of $\f$, since traces are preserved at each rule application, if there exists an infinitely progressing trace along with $(e'_j)_{0\leq j}$, there also exists the corresponding infinitely progressing trace along with $(e_i)_{0\leq i}$. Note that when $e_i$ and $e_{i+1}$ on $\prminus$ are respectively a bud and its companion, we have $f(e_i) = f(e_{i+1})$, and hence, at this point, the trace does not break.\qed 
\end{proof}

\subsection{Theorem~\ref{thm: adm subst in LK}}
\begin{proof}
Assume that $\prvar$ is a $\CLK$ proof of $\Gamma\vdash\Delta$ s.t. any substitution in $\prsubst(\prvar)$ is an atomic substitution. By Lemma~\ref{lemm: prplus to prvar}, the existence of such a proof $\prvar$ is guaranteed.

By Lemma~\ref{lemm: prd to prminus}, we can construct the pre-proof $\prminus$ of $\Gamma\vdash\Delta$ that does not contain the application of $(\subst)$ and the function $\f$. 
Let $(e_i)_{0\leq i}$ be an infinite path in $\prminus$. 
There is a corresponding path $(e'_j)_{0\leq j}$ in $\prvar$. 
$(e'_j)_{0\leq j}$ is an infinite path since $(e_i)_{0\leq i}$ is an infinite path. Since $\prvar$ is proof, there is an infinitely progressing trace along with $(e'_j)_{0\leq j}$. By Lemma~\ref{lemm: trace from plus to minus}, there is an infinitely progressing trace along with $(e_i)_{0\leq i}$. 
Assuming an arbitrary infinite path on $\prminus$, there exists an infinitely progressing trace. 
Therefore, $\prminus$ satisfies the global trace condition. 
Consequently, $\prminus$ is a proof, and we have established the existence of a proof of $\Gamma\vdash \Delta$ that does not contain the application of $(\subst)$.
\qed
\end{proof}
\end{document}